%% file: dag.tex
\documentclass[submission,copyright,creativecommons]{eptcs}
\providecommand{\event}{FOCLASA 2012} 
\usepackage{breakurl}             
\usepackage{epsfig}
\usepackage{amsfonts}
\usepackage{amssymb,amsthm}
\usepackage{color}
\usepackage[leqno]{amsmath,stmaryrd}
\usepackage{txfonts}
\usepackage{verbatim}
\usepackage{enumerate}
\usepackage{hyperref}
\usepackage{prooftree}
\usepackage{cmll}
\usepackage{syndmeta}
\usepackage{extra}
\usepackage{xypic}
\usepackage{fancybox}
\usepackage{wrapfig}

\newtheorem{theorem}{Theorem}[section]
\newtheorem{definition}[theorem]{Definition}
\newtheorem{proposition}[theorem]{Proposition}
\newtheorem{lemma}[theorem]{Lemma}

\newcommand{\nb}[1]{{\color{blue} #1}}

\begin{document}

\title{
 A Provenance Tracking Model for Data Updates
}

\author{
 Gabriel Ciobanu
 \qquad\quad
 Ross Horne 
\institute{
 Romanian Academy, Institute of Computer Science, \\
 Blvd. Carol I, no. 8, 700505 Ia{\c{s}}i, Romania 
}
 \email{\quad gabriel@info.uaic.ro \qquad ross.horne@gmail.com}
}

\def\titlerunning{A Provenance Tracking Model for Data Updates}
\def\authorrunning{G. Ciobanu \& R. Horne}

\maketitle

\begin{abstract}
\input{abstract}

\end{abstract}

\input{intro}

\input{motivation}

\input{semantics}
\input{process}
\input{denotation}
\input{conc}

\bibliographystyle{eptcs}
\bibliography{biblio}
\end{document}

%% file: abstract.tex
For data-centric systems, provenance tracking is particularly important when the system is open and decentralised, such as the Web of Linked Data.
In this paper, a concise but expressive calculus which models data updates is presented. The calculus is used to provide an operational semantics for a system where data and updates interact concurrently. The operational semantics of the calculus also tracks the provenance of data with respect to updates. This provides a new formal semantics extending provenance diagrams which takes into account the execution of processes in a concurrent setting. Moreover, a sound and complete model for the calculus based on ideals of series-parallel DAGs is provided.
The notion of provenance introduced can be used as a subjective indicator of the quality of data in concurrent interacting systems.

%% file: intro.tex
\section{Introduction}

There is a growing trend to publish data openly on the Web. This movement is gaining significant momentum as the governments of several countries and numerous other organisations adopt common principles for publishing data~\cite{Bizer2009}.
Data published according to these principles is referred to as Linked Data, due to the use of URIs to establish links between published data.
By establishing links between arbitrary data sets, significant problems emerge that are of a different flavour to those associated with traditional closed databases.


Many of the new problems which emerge in this scenario are due to the the decentralised nature of the published data. Some significant challenging problems include: the efficient execution of distributed queries and processes which exploit multiple data sources; the impossibility of enforcing a global schema on data; the lack of boundaries for data ensuring the impossibility of complete results; and establishing global standards for data formats and protocols.

This work considers another essential problem, which reflects the diversity of published data. The challenge considered here is that each piece of data published has a varying degree of trust or relevance. A user may consider data published by the BBC to be more trustworthy than data published on a personal blog. However, if the blog is run by a political activist that the consumer of data approves of, then the blog may be more relevant. Thus data should not be associated with a specific trust measure. Instead, some extra information about the data should be tracked, i.e.~the provenance of data. From the extra information provided by the provenance of data, the consumer may judge the quality of the associated data according to their own policy.

Provenance can track several characteristics of the origin of data. Characteristics include ``who'' has influenced the data, ``where'' the data has been located, and ``how'' the data is produced~\cite{Cheney2009}. For Linked Data, a basic notion of ``where'' provenance called a named graph, which indicates where the data is located now, is the recognised standard~\cite{Carroll2005}. In related work, a model extending named graphs is used to track more comprehensive ``where'' and ``who'' provenance~\cite{Dezani2012}. The related work associates trees of identifiers for agents and locations with data. This allows a history of where the data has been published and who published it to be tracked.

This work focuses on a form of ``how" provenance. This form of ``how" provenance tracks causal relationships between stored data and data that was used to produce the data~\cite{Cheney2010}.
For instance, due to a change in usage of a building, data about the building may be updated. The updates may replace or extract information from the original data. Thus ``how" provenance can be used to determine how old data influenced the new data with respect to an update.

The notion of ``how'' provenance investigated is strongly related to event based models of causality~\cite{Boudol1989,Pratt1986}.
This model clarifies, for the first time, the relationship between concurrent process and the provenance diagrams that they produce.
An operational semantics formalises the operational behaviour of processes while recording the provenance associated with the resulting data.
The model presented provides insight that may be used to refine the definition of provenance diagrams. Provenance diagrams that arise from concurrent updates are guaranteed to be in a particular (series-parallel) form. This insight is a contribution to the effort to establish a common notion of a provenance diagram~\cite{Moreau2011}. Furthermore, the model presented is proven to be sound and complete. The formal model provides a foundation for investigating problems associated with tracking and exploiting the provenance of data, including querying provenance~\cite{Buneman2006,Anand2010}, and employing trust metrics~\cite{Green2007}.

%% file: motivation.tex
\section{Causal Dependencies in Provenance Diagrams}
\label{section:motivation}

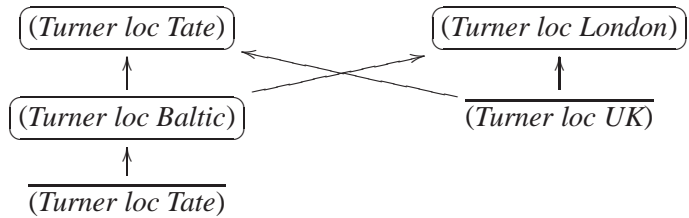
\begin{figure}[b]
\[
\xymatrix@C=2pc@R=1pc{
 \ovalb{\triple{\Turner}{\location}{\Tate}}
 &\quad&
 \ovalb{\triple{\Turner}{\location}{\London}}
 \\
 \ovalb{\triple{\Turner}{\location}{\Baltic}} \ar[u]\ar[urr]
 &\quad&
 \co{\triple{\Turner}{\location}{\UK}} \ar[u]\ar[ull]
 \\
 \co{\triple{\Turner}{\location}{\Tate}} \ar[u]
}
\]
\caption{The Turner Prize is held at the Tate Britain in London. However, in 2011 it was held in The Baltic Gallery in Gateshead, but returned to the Tate Britain in 2012. The data in the above diagram is about the location of the Turner Prize. Edges are causal relationships indicating the data consumed to produce new data as the location of the Turner Prize is updated.}
\label{figure:recipe}
\end{figure}

This work focusses on a particular aspect of provenance tracking. The aspect considered is a form of ``how" provenance, which indicates how old data contributed to producing new data.
The consensus in the provenance community is that provenance diagrams which record this information form a directed acyclic graph (DAG), where the edges are transitively closed.
A standard format for representing provenance, called the Open Provenance Model~\cite{Moreau2011}, encompasses this notion of provenance.
The informal definitions provided by the standard are as follows.

For this work, \textit{artefacts} are data tuples.
The \textit{was derived from} relation between artefacts is such that if there is an edge from artefact $\annotate{d_2}$ to artefact $\annotate{d_1}$, then there is a causal relationship that indicates that $\annotate{d_1}$ needs to have been generated to enable $\annotate{d_2}$ to be generated.
The standard defines a multi-step was derived from relation. This is simply the transitive closure of the \textit{was derived from} relation, indicating that an artefact had an influence on another artefact.

A provenance diagram that indicates the provenance of two stored pieces of data, where the stored data is indicated by an over line,
 is presented in Fig.~\ref{figure:recipe}.
The example is used throughout this work and concerns monuments adjacent to the venue of the workshop.

%% file: semantics.tex
\section{A Syntax and Semantics for Provenance Tracking Data Updates}
\label{section:semantics}

This section introduces the syntax and semantics for a concurrent interacting system that tracks provenance.
The provenance community have introduced provenance structures based on DAGs; therefore a process model which gives rise to DAGs is considered~\cite{Moreau2011}. Unfortunately, many models of concurrency are based on traces or trees rather than DAGs, such as in the calculus and provenance structures introduced in~\cite{Souilah2009}. This limitation is addressed by providing a non-interleaving semantics, inspired by~\cite{Gischer1988}.

\subsection{An Abstract Syntax for Processes}

The grammar for processes is provided in Fig.~\ref{figure:syntax}. The concepts are summarised and made precise by the operational and denotational semantics presented in this work.

\begin{figure}[hbt]
\begin{gather*}
\begin{array}{c}
a\quad\mbox{name}
\qquad
x\quad\mbox{variable}
\\[20pt]
\lambda \Coloneqq x \mid a \quad\mbox{variable or name}
\\[20pt]
\tuple \Coloneqq \lambda \mid \lambda\lambda \mid \lambda\lambda\lambda \mid \hdots \quad\mbox{data tuple}
\end{array}
\qquad\qquad
\begin{array}{rlr}
\Proc \Coloneqq & \cunit                 & \mbox{skip} \\
       \mid & \tuple                       & \mbox{consume data} \\
       \mid & \co{\tuple}                  & \mbox{stored data} \\
       \mid & \annotate \tuple                & \mbox{artefact} \\
       \mid & \Proc \cthen \Proc            & \mbox{seq} \\
       \mid & \Proc \conc \Proc               & \mbox{par} \\
       \mid & \Proc \oplus \Proc              & \mbox{choose} \\
       \mid & \cexists{x} \Proc           & \mbox{exists} \\
\end{array}
\end{gather*}
\caption{The syntax of processes.}
\label{figure:syntax}
\end{figure}

\paragraph{The data tuples.}

The basic unit of information considered in this work is a tuple of names. Tuples are commonly used to convey data.
Linked Data is based on RDF which involves triples of names~\cite{Bizer2009,Horne2011}. RDF makes use of URIs for names, since URIs provide a globally recognised naming system. In Linked Data, often RDF triples are extended to quadruples of URIs where the extra URI indicates where the triple is located~\cite{Carroll2005}. This provides a basic notion of ``where'' provenance. This notion of ``where'' provenance is extended in~\cite{Dezani2012}.


\paragraph{The artefacts.}

A data tuple can be stored, represented as $\co{d}$. Stored data can then be consumed in an interaction with the process $d$.
The result is an artefact $\annotate{d}$ used to explicitly track interactions which have occurred. An artefact is used to record a data tuple involved in an interaction. Artefacts are used to capture ``how'' provenance.


\paragraph{The multiplicatives.}

There are two multiplicative operators. The ``par'' multiplicative represents the parallel composition of processes where interactions between processes are permitted. The ``seq'' multiplicative represents the strict sequential composition of processes, where the first process must terminate before the second process begins, hence the second process is causally dependent on the first process.
There is one unit for the multiplicatives, namely skip, which represents a successful action with no side effects.

\paragraph{The additives.}

There are two additives: $\oplus$ represents a choice between two branches; $\exists$ represents a choice between all possible name substitutions for the bound variable.

\subsection{Operational Semantics of Processes}

Deductive systems are typically presented using inference rules applied at the base of a syntax tree, as in the sequent calculus.
However, such systems are unsuited to systems which mix commutative and non-commutative operators~\cite{Guglielmi2007}. For this reason, a deep inference style of presentation is adopted, where inference rules can be applied at an arbitrary depth in a formula.

\begin{figure}
\begin{gather*}
\cunit \cthen Q \equiv Q \cthen \cunit \equiv Q \conc \cunit \equiv Q
\qquad
P \conc (Q \conc R) \equiv (P \conc Q) \conc R
\qquad
P \cthen (Q \cthen R) \equiv (P \cthen Q) \cthen R
\qquad
Q \conc R \equiv R \conc Q
\\[12pt]
\qquad
\left(P \oplus Q\right) \oplus R \equiv P \oplus \left(Q \oplus R\right)
\qquad\quad
P \oplus Q \equiv Q \oplus P
\qquad\quad
P \oplus P \equiv P
\\[12pt]
\left(P \oplus Q\right) \cthen R \equiv \left(P \cthen R\right) \oplus \left(Q \cthen R\right)
\qquad
P \cthen \left(Q \oplus R\right) \equiv \left(P \cthen Q\right) \oplus \left(P \cthen R\right)
\qquad
\left(P \oplus Q\right) \conc R \equiv \left(P \conc R\right) \oplus \left(Q \conc R\right)
\\[12pt]
\cexists{x} \left(P \oplus Q\right) \equiv \cexists{x} P \oplus \cexists{x} Q
\qquad\qquad
\cexists{x} \cunit \equiv \cunit
\\[12pt]
\cexists{x} \left(P \conc S\right) \equiv \cexists{x} P \conc S
\qquad\quad
\cexists{x} \left(S \cthen Q\right) \equiv S \cthen \cexists{x} Q
\qquad\quad
\cexists{x} \left(P \cthen S\right) \equiv \cexists{x} P \cthen S
\\
\mbox{where $S$ is a process where $x$ does not appear free}
\end{gather*}
\caption{The structural congruence, which can be applied at any point in a derivation.}
\label{figure:struc}
\end{figure}

A structural congruence which extends $\alpha$-conversion is introduced, in Fig.~\ref{figure:struc}, which is used to rearrange processes. The structural congruence ensures that the order of composition matters for sequential composition, but does not matter for parallel composition. For simplicity, both parallel composition and sequential composition share the same unit. The structural congruence handles contraction for choice, using idempotency. The other rules of the structural congruence determine how operators distribute over each other. Distributivity properties are used in related models of concurrency~\cite{Gischer1988,Hoare2011}. Note that this selection of rules is not minimal; however they are used in Sec.~\ref{section:denotation} to rewrite processes into normal forms, thereby simplifying the completeness proof.

A deductive system is presented in Fig.~\ref{figure:deduction}. Deductions may be applied at any depth in a process, as with the structural congruence. Deductions are presented with the premise above the line and the conclusion below the line.

\begin{figure}[h]
\begin{gather*}
\begin{prooftree}
{
 \annotate{d}
}
\justifies
{
 \co{d} \conc d
}
\using
\textit{interact}
\end{prooftree}
\qquad
\begin{prooftree}
{
 \left(P \conc Q\right)
 \cthen
 \left(P' \conc Q'\right)
}
\justifies
{
 \left(P \cthen P'\right)
 \conc
 \left(Q \cthen Q'\right)
}
\using
\textit{sequence}
\end{prooftree}
\qquad
\begin{prooftree}
{P}
\justifies
{
 P \oplus Q
}
\using
\textit{choice}
\end{prooftree}
\qquad
\begin{prooftree}
{P\sub{x}{a}}
\justifies
{
 \exists{x} P
}
\using
\textit{exists}
\end{prooftree}
\end{gather*}
\caption{The deductive system for processes. All deductions can be applied in any context.}
\label{figure:deduction}
\end{figure}

\paragraph{The interact rule.}

The interact rule only applies to tuples. The rule indicates that a stored tuple is consumed by the process which deletes that triple. The result of the interaction is an artefact that records the consumed tuple.

\paragraph{The sequence rule.}

The sequence rule reorders processes composed in parallel. The premise is more deterministic than the conclusion. The premise decides which part of the process will execute first; whereas the conclusion leave open several other opportunities. This rule allows parts of a process to travel to the intended location where they will interact. This rule appears in related models of concurrency~\cite{Gischer1988,Guglielmi2007,Prisacariu2010}.

\paragraph{The additives.}

The premises of the additives indicate the branch that is chosen. For choice, either the left or the right branch is chosen. For exists, any name may be substituted for the bound variable. This kind of choice is known as external choice in process calculi, where exists is an infinite external choice.

%% file: process.tex
\section{A Process Calculus for Provenance Tracking Updates}
\label{section:processes}

This section identifies a sub-grammar of processes that model certain systems. The systems modelled are those which involve stored data composed in parallel with updates. The updates involve the removal of some stored data satisfying a query, followed by the insertion of some new stored data.

The operational semantics for processes are provided by the rules of the system in the previous section. A system can evolve to a given state if and only if the new state entails the original state. Notice that implication is in the opposite direction to the evolution of the system. The direction of implication is in line with related approaches to operational semantics~\cite{Kobayashi1993}.

\begin{figure}[h]
\begin{gather*}
\begin{array}{rl}
\Data \Coloneqq & \cunit \\
           \mid & \co{d} \\
           \mid & \Data \conc \Data
\end{array}
\qquad\qquad
\begin{array}{rl}
\Update \Coloneqq & \Query \cthen \Data    \\
             \mid & \Update \oplus \Update    \\
             \mid & \cexists{x} \Update   \\
\end{array}
\\[10pt]
\begin{array}{rl}
\Query \Coloneqq & \cunit \\
            \mid & d     \\
            \mid & \Query \conc \Query   \\
            \mid & \Query \oplus \Query \\
            \mid & \cexists{x} \Query   \\
\end{array}
\qquad\qquad
\begin{array}{rl}
\System \Coloneqq & \cunit         \\
              \mid & \annotate d      \\
              \mid & \Update       \\
              \mid & \co{d}        \\
              \mid & \System \cthen \System  \\
              \mid & \System \conc \System     \\
\end{array}
\end{gather*}
\caption{Sub-algebras of processes for data, queries, updates and systems.}
\end{figure}

\paragraph{Data.} 

Data simply represents zero or more stored data atoms. The following presents two stored triples in RDF format, which consist of three URIs: the subject, property, and object.
\[
\co{
 \triple{\Sage}{\type}{\Hall}
}
\conc
\co{
 \triple{\Baltic}{\type}{\Gallery}
}
\]
Note that all names are active URIs which link to real published Linked Data. The reader is invited to follow the URIs to witness the examples in a real context.

\paragraph{Queries.}

Parallel composition $\conc$ and choice $\oplus$ are exploited to model the following queries. As in~\cite{Horne2011}, the existential quantifier is used to select URIs which occur in data. The following pattern uses choice to select between two objects. This example discovers a concert hall located in either Newcastle or Gateshead.
\[
\cexists{x}
\left(
 \triple{x}{\type}{\Hall}
 \conc
 \left(
  \triple{x}{\location}{\Newcastle}
  \oplus
  \triple{x}{\location}{\Gateshead}
 \right)
\right)
\]
Note that a tighter operational semantics could be provided by using a tensor product to join queries~\cite{Horne2011}. A tensor product ensures both parts of a query are answered atomically. Unfortunately, the calculus for Linked Data in~\cite{Horne2011} has an interleaving semantics, which would give rise to trees of provenance diagrams as in~\cite{Souilah2009}.
Future work would be to combine the strengths of both calculi.

\paragraph{Updates.}

The following is an example of an update which applies to some stored data.
The existential quantification discovers a name which is used in the delete statement and the data stored after the delete.
The Baltic Art Gallery is a converted flour mill. The update turns a depiction of the old flour mill into a depiction of the new art gallery.
\[
\begin{array}{l}
\co{\triple{\Mill}{\depiction}{\photo}}
\conc 
\cexists{x}
\left(
 \triple{\Mill}{\depiction}{x}
 \cthen
 \co{
  \triple{\Baltic}{\depiction}{x}
 }
\right)
\end{array}
\]
The above process is provable from the following process, using the exists, sequence and interact rules. This means that the system above can evolve to the system below.
\[
\begin{array}{l}
\annotate{\triple{\Mill}{\depiction}{\photo}}
\cthen
\co{
 \triple{\Baltic}{\depiction}{\photo}
}
\end{array}
\]
Notice that the original stored triple appears as an artefact, which the new 
triple is dependent on. This provides ``how" provenance that indicates the old triple used to create the new triple.

\paragraph{Distinctions between execution paths.}

There are multiple ways of evaluating processes. Different methods of evaluation can give rise a different provenance. Here three distinct executions of the same process are presented to demonstrate the complexity of provenance tracking in a concurrent setting.

An example which involves two updates executed in parallel is presented below. 
It is a common misconception that The Sage and Baltic Art Gallery are prominent monuments in Newcastle. In reality they are located in Gateshead on the opposite bank of the river Tyne\footnote{Indeed the venue of FOCLASA 2012 is also in Gateshead, rather than in Newcastle.}.
The updates transform the location of these monuments from Newcastle to Gateshead.
\[
\begin{array}{l}
\co{
 \triple{\Sage}{\location}{\Gateshead}
}
\conc
\left(
 \triple{\Sage}{\location}{\Gateshead}
 \cthen
 \co{
  \triple{\Sage}{\location}{\Newcastle}
 }
\right)
\conc
\\
\co{
 \triple{\Baltic}{\location}{\Gateshead}
}
\conc
\left(
 \triple{\Baltic}{\location}{\Gateshead}
 \cthen
 \co{
  \triple{\Baltic}{\location}{\Newcastle}
 }
\right)
\end{array}
\]
The process below yields the process above, using the sequence rule.
The two updates occur independently, hence each provenance is independent.
\[
\begin{array}{l}
\left(
 \annotate{\triple{\Sage}{\location}{\Gateshead}}
 \cthen
 \co{
  \triple{\Sage}{\location}{\Newcastle}
 }
\right)
\conc
\\
\left(
 \annotate{\triple{\Baltic}{\location}{\Gateshead}}
 \cthen
 \co{
  \triple{\Baltic}{\location}{\Newcastle}
 }
\right)
\end{array}
\]
The process below yields both process above. This suggest that the two updates were combined before they were applied, hence data produced by each update is dependent on the artefact of the other update. Therefore the process below has stronger dependencies than the process above.
\[
\begin{array}{l}
\left(
 \annotate{\triple{\Sage}{\location}{\Gateshead}}
 \conc
 \annotate{\triple{\Baltic}{\location}{\Gateshead}}
\right)
\cthen
\left(
 \co{
  \triple{\Sage}{\location}{\Newcastle}
 }
 \conc
 \co{
  \triple{\Baltic}{\location}{\Newcastle}
 }
\right)
\end{array}
\]
Indeed the above process can be refined further to impose a sequential dependency on the artefacts. Thus the execution of the concurrent processes greatly affects the form of ``how" provenance.

\paragraph{The Turner Prize revisited.}

The operational behaviour which gives rise to the provenance diagram in the introduction can now be expressed. 
The initial configuration is expressed below. It shows two stored triples, an update that moves the exhibition from the Tate Britain to the Baltic and broadens London to the UK, and an update which moves the exhibition back from the Baltic to the Tate Britain.
\[
 \begin{array}{l}
   \co{\triple{\Turner}{\location}{\London}} \conc \co{\triple{\Turner}{\location}{\Tate}}
   \conc \\
  \left(
   \triple{\Turner}{\location}{\Tate}  \conc \triple{\Turner}{\location}{\London}
  \right)
  \cthen
  \left( \co{\triple{\Turner}{\location}{\Baltic}} \conc \co{\triple{\Turner}{\location}{\UK}} \right)
 \conc \\
  \left(
   \triple{\Turner}{\location}{\Baltic}
   \cthen
   \co{\triple{\Turner}{\location}{\Tate}}
  \right)
 \end{array}
\]
By applying the sequence rule several times the processes can be rearranged as follows.
\[
\begin{array}{l}
  \left(
   \triple{\Turner}{\location}{\Tate} \conc \co{\triple{\Turner}{\location}{\Tate}}
   \conc
   \co{\triple{\Turner}{\location}{\London}} \conc \triple{\Turner}{\location}{\London}
  \right)
  \cthen \\
 \left(
  \left(
   \left(\co{\triple{\Turner}{\location}{\Baltic}} \conc \triple{\Turner}{\location}{\Baltic}\right)
   \cthen
   \co{\triple{\Turner}{\location}{\Tate}}
  \right)
  \conc
  \co{\triple{\Turner}{\location}{\UK}}
 \right)
 \end{array}
\]
Finally, by applying the interact rule the delete operations and stored data cancel each other out. The interaction produce the artefacts that record the provenance of the data.
\[
 \begin{array}{l}
  \left(
   \annotate{\triple{\Turner}{\location}{\Tate}}
   \conc
   \annotate{\triple{\Turner}{\location}{\London}}
  \right)
  \cthen \\
 \left(
  \left(
   \annotate{\triple{\Turner}{\location}{\Baltic}} \cthen
   \co{\triple{\Turner}{\location}{\Tate}}
  \right)
  \conc
  \co{\triple{\Turner}{\location}{\UK}}
 \right)
 \end{array}
\]
The next section provides a denotational semantics where the denotation of above process is exactly the provenance diagram in the introduction.

%% file: denotation.tex
\section{A Denotational Semantics for the Provenance Tracking Calculus}
\label{section:denotation}

This section provides a denotational semantics for the calculus. A denotational semantics provides a sound and complete model which increases confidence in the definition of the calculus. In this case, the semantics of the calculus fulfils an additional purpose. It also makes explicit the connection between certain terms of the calculus and provenance diagrams. Furthermore, a restriction on provenance diagrams that track series-parallel computations is highlighted.

The denotational semantics, similarly to provenance diagrams, is based on directed acyclic graphs (DAGs). The denotation relies on some technical apparatus. Firstly, DAGs are restricted by a forbidden minor property, which guarantees that each DAG arises from applying series and parallel composition to smaller DAGs. Secondly, homomorphisms between DAGs are defined such that the inference rules of the calculus hold. By taking ideals of series-parallel DAGs with respect to these homomorphism, a sound and complete model is obtained.

\subsection{Series-Parallel DAGs and the $N$-free Condition}

This section recalls some standard definitions which are used to build a denotational semantics. The definition of a DAG is standard, as are the definitions of the transitive closure of a graph and the notion of a graph homomorphism. Transitive DAGs are used because provenance diagrams are transitive, and graph homomorphism are used to compare the structure of such diagrams. 
\begin{definition}
A DAG $D = \left(V, E\right)$ is a digraph with no directed cycles. 
Let $A = \left(V, E\right)$ and $B = \left(V', E'\right)$ be graphs. A graph homomorphism is given by a function on vertices $f \colon V \rightarrow V'$ such that if $\left(u,v\right) \in E$ then $\left(f(u),f(v)\right) \in E'$.
Two graphs are isomorphic iff there exists a bijective homomorphism whose inverse function is also a homomorphism.
A transitive digraph is such that if there exists a path from $u$ to $v$, then there exists an edge from $u$ to $v$.
A transitive closure of a digraph $\left(V,E\right)$ is a minimal transitive digraph $\left(V,E'\right)$ such that there exists an injective graph homomorphism from $\left(V,E\right)$ to $\left(V,E'\right)$.
A graph $\left(V,E\right)$ is a sub-graph $\left(V',E'\right)$ if and only if $V \subseteq V'$ and $E = E' \cap V \times V$.
\end{definition}
Several series-parallel digraphs are studied in~\cite{Valdes1979}. Here transitive series-parallel DAGs are defined. The series-parallel restriction on transitive DAGs is required because this work considers provenance diagrams which arise from the execution of series-parallel processes.
%
\begin{definition}
\label{definition:series-parallel}
The trivial DAG with no vertices is a series-parallel DAG, and the DAG with a single vertex and no edges is a series-parallel DAG. If $G_0 = \left(V_0, E_0 \right)$ and $G_1 = \left(V_1, E_1 \right)$ are series-parallel graphs with disjoint vertices, then the following are series-parallel DAGs.
\begin{itemize}
\item $G_0 \pipar G_1$ defined by $\left(V_0 \cup V_1, E_0 \cup E_1 \right)$.
\item $G_0 \cthen G_1$ defined as the transitive closure of $\left( V_0 \cup V_1, E_0 \cup E_1 \cup \left(L \times R\right)\right)$, where $L$ is the source nodes of $G_0$ and $R$ is the sink nodes of $G_1$.
\end{itemize}
\end{definition}
In structural graph theory it is studied how graph classes either can be defined by forbidden minors, or by being glued together from simple starting graphs (as in the definition above). 
A forbidden minor is a sub-graph with a particular form; the forbidden minor for series-parallel 
DAGs has an $N$-shape, as proven in~\cite{Valdes1979}.

\begin{theorem}[Forbidden minor]
\label{theorem:valdes}
A transitive DAG is series-parallel if and only if it does not have a sub-graph isomorphic to
$N = \left(\left\{v_0, v_1, v_2, v_3 \right\}, \left\{(v_2, v_0),(v_3,v_0),(v_3,v_1)\right\}\right)$.
\end{theorem}

Notice that use of transitive DAGs is motivated, by provenance diagrams; while the series-parallel restriction is motivated by concurrent processes. Thus the model studies structures which respect both provenance and the processes which track the provenance.


\subsection{Interacting Series-Parallel DAGs Labelled with Data}

The notion of a series-parallel DAG is extended with labels. The labels allow data to be accommodated in the model. Also the notion of a homomorphism is extended to allow interactions between data and operations on data which give rise to artefacts.

The definition of a labelled graph is standard. A special kind of homomorphism is defined on labelled DAGs. This smoothing homomorphism is bijective, but does not define an isomorphism. Thus vertices are preserved, but extra edges may appear. 

\begin{definition}
Fix $\Sigma$ as the set of labels which are either tuples $d$, stored tuples $\co{d}$ or artefacts $\annotate{d}$. 
A labelled graph $\left(V, E, \mu\right)$ is such that $\left(V, E\right)$ is a graph and $\mu : V \rightarrow \Sigma$ is a labelling function from vertices to labels. 
Let $A = \left(V, E, \mu\right)$ and $B = \left(V', E', \mu'\right)$ be labelled DAGs.
A labelled homomorphism $f$ from $A$ to $B$ is such that $f$ is a graph homomorphism from $\left(V, E\right)$ to $\left(V', E'\right)$ and for all vertices $u$, $\mu(u) = \mu'(f(u))$. A smoothing homomorphism is a bijective labelled homomorphism.
\end{definition}

The notation of a smoothing homomorphism defined above is used to characterise the sequence rule. To capture both the sequence rule and the interact rule, interaction homomorphisms are introduced. The definition involves a coherence condition which captures the conditions under which an interaction may occur. Two vertices can interact if they have complementary labels and they are in parallel with each other. This leads to the following definition.
\begin{definition}
For a labelled graph $A = \left(V, E, \mu\right)$, define $u \inter_d v$ in $A$ such that there is no directed path between $u$ and $v$, and either $\co{d} = \mu(u) = \co{\mu(v)}$ or $\co{\mu(u)} = \mu(v) = \co{d}$. Let $A = \left(V, E, \mu\right)$ and $B = \left(V', E', \mu'\right)$ be labelled DAGs. An interaction homomorphism $f$ from $A$ to $B$ is a labelled graph homomorphism such that $f$ is onto and, if $f(u) = f(v)$, one of the following hold: either $u \inter_d v$ in $A$ and $\mu'(f(u)) = \annotate{d}$; or $u = v$ and $\mu(u) = \mu'\left(f(u)\right)$.
\end{definition}


The following example demonstrates two compatible vertices mapped to the same vertex by an interaction homomorphism.
\[
\vcenter{
\xymatrix@C=1pc@R=0.5pc{
 a&b&\co{b} \\
 & d \ar[ul]\ar[u] & \co{d} \ar[u] \\
 & \co{b} \ar[u] & \co{a}\ar[u]
}
}
\qquad
\longrightarrow
\qquad
\vcenter{
\xymatrix@C=1pc@R=0.5pc{
 a&b&\co{b} \\
 & \annotate{d} \ar[ul]\ar[u]\ar[ur] \\
 & \co{b}\ar[u] & \co{a}\ar[lu]
}
}
\]
Note that the diagrams in examples represent equivalence classes of labelled graphs up to labelled graph isomorphism. Thus only the labels and not the underlying vertices are indicated. The same practice is followed when presenting provenance diagrams.

The homomorphisms defined over labelled DAGs are used to generate ideals. Ideals are sets of labelled series-parallel DAGs closed with respect to either smoothing or interacting homomorphisms.

\begin{definition}
A smoothing/interacting ideal $I$ is a set of labelled series-parallel DAGs such that if $A \in I$ and there exists a smoothing/interacting homomorphism $f : A \rightarrow B$, then $B \in I$.
For any set of labelled series-parallel DAGs $P$ the smoothing/interacting ideal closure of $P$, denoted $\iota_{s} P$/$\iota_{i} P$, is the least ideal containing $P$, defined as the intersection of all smoothing/interacting ideals $I$ such that $P \subseteq I$.
\end{definition}

These ideals are employed to denote processes in the next section. Ideal closure is essential for the denotation of parallel composition.

\subsection{Correctness of the Denotational Semantics}


The denotational semantics for processes is defined using the ideals introduced in the previous section. Most operations on ideals are the obvious point-wise extension of the corresponding operator. The main subtlety is that parallel composition introduces new possibilities for both smoothing and interaction, which are not represented by the point-wise parallel composition of ideals. Thus the ideal closure is employed to denote parallel composition. Valuations are used to represent substitutions, which are required to denote existential quantification.

\begin{definition}[denotation]
A valuation $v$ is a mapping from variables to names. Let $v \mathclose{\left[x \mapsto a\right]}$ be the valuation which is the same as $v$ except at $x$ where it maps to $a$. The effect of a valuation on a label is defined as follows.
\begin{gather*}
\valu{v}{\annotate{d}} = \annotate{\valu{v}{d}}
\qquad
\valu{v}{\left(\co{d}\right)} = \co{\valu{v}{d}}
\qquad
\valu{v}{\left(\lambda_0..\lambda_n\right)} = \valu{v}{\lambda_0}..\valu{v}{\lambda_n}
\qquad
\valu{v}{a} = a
\qquad
\valu{v}{x} = v(x)
\end{gather*}
The denotation of a process with respect to a valuation $v$ satisfies the following, where $h \in \left\{ s,i \right\}$, $\epsilon$ is the set containing the empty labelled graph, and $\element{l,v}$ is the equivalence class of labelled graph with one vertex labelled with $l^v$ with respect to labelling isomorphism.
\begin{gather*}
\denotation{v}{\cunit}_h = \epsilon
\qquad\quad
\denotation{v}{l}_h = \element{l,v}
\qquad\quad
\denotation{v}{\exists{x} P}_h = \bigcup_{a \in \names} \denotation{v \mathclose{\left[x \mapsto a\right]}}{P}_h
\\[10pt]
\denotation{v}{P \oplus Q}_h = \denotation{v}{P}_h \cup \denotation{v}{Q}_h
\qquad\qquad
\denotation{v}{P \cthen Q}_h = \left\{ A \cthen B \mid (A,B) \in \denotation{v}{P}_h \times \denotation{v}{Q}_h \right\}
\\[12pt]
\denotation{v}{P \pipar Q}_h = \iota_h \left\{ A \pipar B \mid (A,B) \in \denotation{v}{P}_h \times \denotation{v}{Q}_h \right\}
\end{gather*}
\end{definition}

All the operations used in the denotational semantics preserve ideals, as verified by the following proposition.
Therefore the denotational semantics is a well defined mapping from processes to ideals.
\begin{proposition}
The following are ideals: $\epsilon$, $\element{l,v}$, the union and intersection of sets of ideals, and the point-wise sequential composition of ideals.
\end{proposition}

Soundness of the calculus defined in Sec.~\ref{section:semantics} with respect to the denotation is straight forward. The proof follows from checking that all equations of the structural congruence hold as set equality of ideals, and that all deductive rules hold as set inclusions of ideals.

\begin{theorem}[soundness]
If $P$ yields $Q$, then, $\denotation{v}{P}_i \subseteq \denotation{v}{Q}_i$ for all valuations $v$.
\end{theorem}
%

Completeness of the calculus with respect to the denotation is more challenging. The proof follows from interpolation lemmas. An interpolation lemma establishes that if there is a strict inclusion between the denotation of processes then there must be a finite sequence of deductions that can be applied to transform one process into the other process. The trick is to rewrite processes into a normal form and deal with each deductive rule one by one.

Firstly consider series-parallel terms, which are processes which does not feature any choice or exists. Two interpolation lemmas apply to series-parallel terms. The first interpolation lemma, stated below, deals only with the sequence rule. This lemma is closely related to the interpolation lemma established in~\cite{Gischer1988}, where a similar calculus without interactions is considered. Thus only smoothing ideals are treated.

\begin{lemma}[sequence interpolation]
\label{lemma:gischer}
Given two series-parallel terms $P$ and $Q$, if $\gischer{v}{P} \subseteq \gischer{v}{Q}$ for all valuations $v$, then either: $\gischer{v}{P} = \gischer{v}{Q}$ for all valuations $v$; or there exists $R$ such that $\gischer{v}{P} \subset \gischer{v}{R} \subseteq \gischer{v}{Q}$ for all valuations $v$, and $P$ yields $R$ is provable using only the sequence rule.
\end{lemma}

The above result is extended to interacting homomorphism in the following interpolation lemma. The proof of this lemma is an important technical contribution of this work. It shows that, for any strict inclusion between the denotation of series-parallel process, either the sequence rule or the interact rule can be applied.
\begin{lemma}[interaction interpolation]
\label{lemma:interaction}
Given two series-parallel terms $P$ and $Q$, if $\dinter{v}{P} \subseteq \dinter{v}{Q}$ for all valuations $v$, then: either $\gischer{v}{P} \subseteq \gischer{v}{Q}$ for all valuations $v$; or there exists $R$ such that $\dinter{v}{P} \subset \dinter{v}{R} \subseteq \dinter{v}{Q}$ for all valuations $v$ and $P$ yields $R$ is provable using only the interact rule.
\end{lemma}
\begin{proof}
\input{proof_interpolation}
\end{proof}

To clarify the significance of the interpolation lemmas consider the running example. The initial configuration of processes is denoted by the following DAG ($D_1$).
\[
D_1:
\xymatrix@C=1pc@R=2pc{
 \co{\triple{\Turner}{\location}{\Tate}}
 &
 {\triple{\Turner}{\location}{\Tate}}
 &
 {\triple{\Turner}{\location}{\London}}
 &
 {\triple{\Turner}{\location}{\Baltic}}
 \\
 \co{\triple{\Turner}{\location}{\London}}
 &
 \co{\triple{\Turner}{\location}{\Baltic}} \ar[u]\ar[ur]
 &
 \co{\triple{\Turner}{\location}{\UK}} \ar[u]\ar[ul]
 &
 \co{\triple{\Turner}{\location}{\Tate}} \ar[u]
}
\]
There exists an interaction homomorphism from $D_1$ to the DAG $D_0$ below, which appears also in Sec.~\ref{section:motivation}.
\[
D_0:
\xymatrix@C=2pc@R=1.5pc{
 \ovalb{\triple{\Turner}{\location}{\Tate}}
 &
 \ovalb{\triple{\Turner}{\location}{\London}}
 \\
 \ovalb{\triple{\Turner}{\location}{\Baltic}} \ar[u]\ar[ur]
&
 \co{\triple{\Turner}{\location}{\UK}} \ar[u]\ar[ul]
 \\
 \co{\triple{\Turner}{\location}{\Tate}} \ar[u]
 &
}
\]
Now, by applying Lemma~\ref{lemma:interaction} three times, we can construct a series of DAGs cumulating in $D_2$ presented below, such that the following properties hold. There exist interaction homomorphisms from $D_1$ to $D_2$ and from $D_2$ to $D_0$, and the process denoted by $D_0$ yields the process denoted by $D_2$ using the interact rule three times. Furthermore, the homomorphism from $D_1$ to $D_2$ is a smoothing homomorphism. Hence, by Lemma~\ref{lemma:gischer}, the process denoted by $D_2$ can be transformed using the sequence rule applied a finite number of times into the process denoted by $D_1$.
\[
D_2:
\xymatrix@C=2pc@R=2pc{
 \co{\triple{\Turner}{\location}{\Tate}}
 &
 {\triple{\Turner}{\location}{\Tate}}
 &
 \co{\triple{\Turner}{\location}{\London}}
 &
 {\triple{\Turner}{\location}{\London}}
 \\
 \co{\triple{\Turner}{\location}{\Baltic}} \ar[u]\ar[ur]\ar[urr]\ar[urrr]
 &
 {\triple{\Turner}{\location}{\Baltic}}\ar[ul]\ar[u]\ar[ur]\ar[urr]
 &
 \co{\triple{\Turner}{\location}{\UK}} \ar[ull]\ar[ul]\ar[u]\ar[ur]
 \\
 &
 \co{\triple{\Turner}{\location}{\Tate}} \ar[u]\ar[ul]
}
\]
Thereby the existence of the interaction homomorphism between $D_1$ and $D_0$ guarantees the existence of a deduction from the process denoted by $D_0$ to the process denoted by $D_1$ using the interact and sequence rules. Indeed the processes and deductions are presented in the example at the end of Sec.~\ref{section:processes}, where the first process in the example is denoted by $D_1$, the second by $D_2$ and the third by $D_0$.

Every process can be written in a normal form, using the structural congruence, as a sum of series-parallel process with all the existential quantification moved to the front of the process, i.e.~for all $P$ there exist series-parallel processes $A_i$ such that $P \equiv \cexists{\vec{x}} \Sigma_{i\in I} A_i$. It is then easy to show that a finite number of choice and exists rules can be applied to prove any inequality between ideals. This establishes the completeness of the calculus with respect to the denotation, stated as follows.

\begin{theorem}[completeness]
If $\denotation{v}{P}_i \subseteq \denotation{v}{Q}_i$ for all valuations $v$, then $P$ yields $Q$.
\end{theorem}
%

Thus the model based on ideals of labelled series-parallel DAGs is a sound and complete model of processes. The labelled DAGs are inspired by the guidelines provided for provenance diagrams~\cite{Moreau2011}; while, the series-parallel processes are motivated by calculi which model systems which produce provenance diagrams. Hence a formal connection between series-parallel DAGs and processes is established. Specifically, provenance diagrams are the denotation of series-parallel processes consisting of only artefacts and stored data. Hence provenance diagrams are contained within a denotation for a provenance tracking calculus. Due to soundness and completeness of the calculus with respect to the denotation, provenance diagrams can be considered in a new operational language based setting.

%% file: proof_interpolation.tex
Assume that $P$ and $Q$ are series-parallel terms such that $\dinter{v}{P} \subset \dinter{v}{Q}$ for all valuations $v$. Also assume that $\gischer{v}{P} \not\subset \gischer{v}{Q}$ for some valuation $v$. 
Since $P$, $Q$ are series-parallel terms, there exist series-parallel DAGs $D_0 = \left(V_0,E_0,\mu_0\right)$, $D_1 = \left(V_1,E_1,\mu_1\right)$ such that $\iota_i D_0 = \dinter{v}{P}$ and $\iota_i D_1 = \dinter{v}{Q}$. Also, since $\dinter{v}{P} \subset \dinter{v}{Q}$, there exists an interacting homomorphism $f \colon D_1 \rightarrow D_0$.

There must be at least one interaction in the homomorphism $f$ exhibited above, i.e. there exists $m,n \in V_1$ such that $f(m) = f(n)$, $m \inter_{d} n$ and $f(m) = w \in V_0$ such that $\mu_0(w) = \annotate{d}$. Suppose otherwise, then for all $m,n \in V_1$ if $f(m) = f(n)$ then $m = n$, and so $f$ is bijective, since interacting homomorphisms are surjective. Hence $f$ is a smoothing homomorphism from $D_1$ to $D_0$, so $\gischer{v}{P} \subset \gischer{v}{Q}$ contradicting the above assumption.

A DAG $D_2 = \left( V_2, E_2, \mu_2 \right)$ is constructed to differ from $D_0$ only by the interaction exhibited by $f$. Firstly, take $V_0$, remove vertex $w$ and include vertices $m$ and $n$, so $V_2 = V_0 \setminus \left\{ w \right\} \cup \left\{ m, n \right\}$.
Let $E_0 \setminus w$ be the set of edges in $E_0$ without the vertex $w$ and define $E_2 = \left(E_0 \setminus w\right)  \cup \left\{ (x,m) \mid (x,w) \in E_0 \right\} \cup \left\{ (m,x) \mid (w,x) \in E_0 \right\} \cup \left\{ (x,n) \mid (x,w) \in E_0 \right\} \cup \left\{ (n,x) \mid (w,x) \in E_0 \right\}$. Retain all the labels of $\mu_0$ except at $m$ and $n$, so if $x = m$ or $x=n$ then $\mu_2(x) = \mu_1(x)$ and otherwise $\mu_2(x) = \mu_0(x)$.

Construct two homomorphisms from $g \colon D_2 \rightarrow D_0$ and $h \colon D_1 \rightarrow D_2$ as follows.
\[
g(x) = \left\{ 
  \begin{array}{l l}
    f(x) & \quad \text{if $x = m$ or $x = n$}\\
    x & \quad \text{otherwise}\\
  \end{array} \right.
\qquad\quad
h(x) = \left\{ 
  \begin{array}{l l}
    x & \quad \text{if $x = m$ or $x = n$}\\
    f(x) & \quad \text{otherwise}\\
  \end{array} \right.
\]
Clearly $f = g \circ h$. Furthermore, both $g$ and $h$ are interacting homomorphisms by the following arguments.
Check that $g$ is a graph homomorphism, by case analysis. Only one case is presented. By definition of $E_2$, if $(m, x) \in E_2$ then $(m, x) \in \left\{ (m,x) \mid (w,x) \in E_0 \right\}$, thus $(g(m),g(x)) = (w,x) \in E_0$. 
Also check that $g$ is an interaction homomorphism, as follows: If $g(x)=g(y)$ then either $x=y$, or $x=m$ and $y=n$. Clearly, $m$ and $n$ are not connected in $E_2$ and both $\mu_2(m) = \mu_1(m)$ and $\mu_2(n) = \mu_1(n)$ hold, so $m \inter_d n$ in $D_2$ and $\mu_0(f(m)) = \mu_0(w) = \annotate{d}$.
Check that $h$ is a graph homomorphism. Only one case is presented. If $(m, x) \in E_1$ then $(w, f(x)) \in E_0$ since $f$ is a graph homomorphism, thus $(m, f(x)) \in \left\{ (m,x) \mid (w,x) \in E_0 \right\}$ so $(m, f(x)) \in E_2$, by definition.
Now consider when $h(x) = h(y)$ either $x=y$ or $x,y\not\in\left\{m,n\right\}$, hence $f(x) = f(y)$, thus $x \inter_d y$ since $f$ is an interacting homomorphism. Suppose, without loss of generality, that $x = m$ and $y\not\in\left\{m,n\right\}$, so $m = f(x)$, but $m \not \in V_0$ contradicting the definition of $f$. Thus $h$ is an interaction homomorphism.

Furthermore, the constructed DAG, $D_2$, is series-parallel. Suppose otherwise, then there exists an $N$-shape isomorphic to a sub graph of $D_2$. Now consider the image of the $N$-shape under $g$. Either zero or one nodes in the $N$-shape are $m$ or $n$ so the image of the $N$ shape is an $N$-shape in $D_0$. By Theorem~\ref{theorem:valdes}, this contradicts the fact that $D_0$ is series-parallel. Now, suppose that both $m$ and $n$ are in the $N$-shape. Since $m \inter_d n$ in $D_2$, $m$ and $n$ are not connected, so an $N$-shape must be of the form $\left\{(m,x),(n,x),(n,y)\right\}$ or $\left\{(x,m),(x,n),(y,n)\right\}$. However $(m, x) \in D_2$ iff $(w, x) \in D_0$ iff $(n, x) \in D_2$ and $(x, m) \in D_2$ iff $(x, w) \in D_0$ iff $(x, n) \in D_2$, so neither shapes are sub-graphs of $D_2$. Thus $D_2$ is $N$-free, hence by Theorem~\ref{theorem:valdes}, $D_2$ is a series-parallel DAG.

Since, $D_2$ is a series-parallel DAG, there exists a series-parallel term $R$ such that $\dinter{v}{R} = \iota_i D_2$. Since $g: D_2 \rightarrow D_0$ exhibiting an interaction and $h \colon D_1 \rightarrow D_2$, the following inequalities hold $\dinter{v}{P} \subset \dinter{v}{R}$ and $\dinter{v}{R} \subseteq \dinter{v}{Q}$. Since $\mu_2(w) = \annotate{d}$, the sub-term $\annotate{d}$ must appear in the process $P = \context{\annotate{d}}$, for some context $\context{}$. Also, since $m \inter_d n$ and, the edges of $D_0$ differs from those of $D_2$ only in that the edges connected to $w$ in $D_0$ are instead connect to both $m$ and $n$ in $D_2$,
the following holds $R \equiv \context{d \pipar \co{d}}$. Thus the interact rule proves that $P$ yields $R$, as required.

%% file: conc.tex
\section{Conclusion}


Provenance is a key problem in processing data which is particularly important in systems that publish data on the Web, such as the Web of Linked Data~\cite{Dezani2012,Horne2011}. Already certain aspects of provenance are gifted with deep theoretical results~\cite{Green2007}. However, there is no sound and complete model for the aspects of provenance tracking considered in this work: specifically ``how" provenance which indicates causal relationships; and a provenance tracking calculus which produces such diagrams by recording interactions between processes and stored data. The relationship between the diagrams and the calculus is exhibited by providing a sound and complete denotational semantics which contains such provenance diagrams.


The examples presented in this paper illustrate that tracking provenance is 
particularly challenging in a concurrent setting. The causal aspects of data 
provenance are closely related to the operational semantics of the systems 
involved. Hence when considering concurrent systems, models of concurrency provide insight into problems associated with provenance in a concurrent setting. 
For instance, this work demonstrates that provenance diagrams that arise from concurrent interactions form series-parallel DAGs. Consequently, certain graph homomorphism problems, which can be employed to query provenance diagrams, can be solved more efficiently for series-parallel digraphs~\cite{Valdes1979}.
This model is proposed as a foundation for ``how" provenance, which can be applied as a subjective measure of the quality of data.



\paragraph{Acknowledgements.}
We thank the reviewers for feedback that improved the exposition.
The work was supported by a grant of the Romanian National Authority for Scientific Research, CNCS-UEFISCDI, project number PN-II-ID-PCE-2011-3-0919.